\documentclass[11pt]{article}

\usepackage[a4paper,left=20mm,top=20mm,right=20mm,bottom=25mm]{geometry}

\usepackage{amsmath, amsfonts, amssymb, amsthm}
\usepackage{algorithm}
\usepackage[noend]{algorithmic}
\usepackage{algorithm}
\usepackage{authblk}

\usepackage{graphicx}  
\usepackage{mathtools}
\usepackage{subfig}
\usepackage[mathscr]{euscript}

\usepackage[colorlinks]{hyperref}
\hypersetup{
	citecolor=blue,
   linkcolor=blue,
}

\usepackage[font=footnotesize,width=.85\textwidth,labelfont=bf]{caption}

\usepackage{color}

\makeatletter
\def\moverlay{\mathpalette\mov@rlay}
\def\mov@rlay#1#2{\leavevmode\vtop{%
    \baselineskip\z@skip \lineskiplimit-\maxdimen
    \ialign{\hfil$\m@th#1##$\hfil\cr#2\crcr}}}
\newcommand{\charfusion}[3][\mathord]{
  #1{\ifx#1\mathop\vphantom{#2}\fi
    \mathpalette\mov@rlay{#2\cr#3}
  }
  \ifx#1\mathop\expandafter\displaylimits\fi}
\DeclareRobustCommand\bigop[1]{%
  \mathop{\vphantom{\sum}\mathpalette\bigop@{#1}}\slimits@
}
\newcommand{\bigop@}[2]{%
  \vcenter{%
    \sbox\z@{$#1\sum$}%
    \hbox{\resizebox{\ifx#1\displaystyle.9\fi\dimexpr\ht\z@+\dp\z@}{!}{$\m@th#2$}}%
  }%
}
\makeatother

\newcommand{\cupdot}{\charfusion[\mathbin]{\cup}{\cdot}}
\DeclareMathOperator{\bigcupdot}{\charfusion[\mathop]{\bigcup}{\cdot}}
\newcommand{\bigjoin}{\bigop{\Join}}
\DeclareMathOperator{\join}{\Join}
\newcommand{\AX}[1]{\textnormal{#1}}
\DeclareMathOperator{\lca}{lca}
\DeclareMathOperator{\symdiff}{\triangle}
\newtheorem{problem}{Problem}

\newcommand{\hc}{\emph{hc}}

\providecommand{\keywords}[1]{\textbf{\textit{Keywords: }} #1}

\newtheorem{theorem}{Theorem}
\newtheorem{lemma}{Lemma}
\newtheorem{corollary}{Corollary}
\newtheorem{remark}{Remark}

\newcommand{\PROBLEM}[1]{{\tt #1}}

\newtheorem{definition}{Definition}
\newtheorem{fact}{Observation}

\begin{document}

\title{Complexity of Modification Problems for Reciprocal Best Match  Graphs}

\author[2,3]{Marc Hellmuth}
\author[1]{Manuela Gei{\ss}}
\author[1,4,5,6]{Peter F.\ Stadler}

\affil[1]{\footnotesize Bioinformatics Group, Department of Computer Science; and
		    Interdisciplinary Center of Bioinformatics, University of Leipzig, 
			 H{\"a}rtelstra{\ss}e 16-18, D-04107 Leipzig}
\affil[2]{Dpt.\ of Mathematics and Computer Science, University of Greifswald, Walther-
  Rathenau-Strasse 47, D-17487 Greifswald, Germany 	}
\affil[3]{Saarland University, Center for Bioinformatics, Building E 2.1, P.O.\ Box 151150, D-66041 Saarbr{\"u}cken, Germany }
\affil[4]{Max-Planck-Institute for Mathematics in the Sciences, 
  Inselstra{\ss}e 22, D-04103 Leipzig}
\affil[5]{Inst.\ f.\ Theoretical Chemistry, University of Vienna, 
  W{\"a}hringerstra{\ss}e 17, A-1090 Wien, Austria}
\affil[6]{Santa Fe Institute, 1399 Hyde Park Rd., Santa Fe, USA} 
\date{}
\normalsize

\maketitle

\abstract{   
  Reciprocal best match graphs (RBMGs) are vertex colored graphs whose
  vertices represent genes and the colors the species where the genes
  reside. Edges identify pairs of genes that are most closely related with
  respect to an underlying evolutionary tree. In practical applications
  this tree is unknown and the edges of the RBMGs are inferred by
  quantifying sequence similarity. Due to noise in the data, these
  empirically determined graphs in general violate the condition of being a
  ``biologically feasible'' RBMG. Therefore, it is of practical interest in
  computational biology to correct the initial estimate.  Here we consider
  deletion (remove at most $k$ edges) and editing (add or delete at most
  $k$ edges) problems.  We show that the decision version of the deletion
  and editing problem to obtain RBMGs from vertex colored graphs is
  NP-hard.  Using known results for the so-called bicluster editing, we
  show that the RBMG editing problem for $2$-colored graphs is
  fixed-parameter tractable.

  A restricted class of RBMGs appears in the context of orthology
  detection. These are cographs with a specific type of vertex coloring
  known as hierarchical coloring.  We show that the decision problem of
  modifying a vertex-colored graph (either by edge-deletion or editing)
  into an RBMG with cograph structure or, equivalently, to an
  hierarchically colored cograph is NP-complete.
}

\bigskip
\noindent

\sloppy

\section{Introduction}

Graph modification problems ask whether there is a set of at most $k$ edges
to delete or to edit (add or delete) to change an input graph into a graph
conforming to certain structural prerequisites.  In computational biology
graph modification problems typically appear as ways to deal with inaccurate
data and measurement error, for instance in genome assembly
\cite{Gonnella:16} or clustering \cite{Boecker:08}.

Here we consider graph modification problems that arise in the context of
orthology assignment. Two genes found in two distinct species are
orthologous if they arose through the speciation event that also separated
the two species. Biologically, one expects these two genes to have
corresponding function. In contrast, paralogous genes, which arose through
a duplication event, are expected to have related but distinct functions
\cite{Gabaldon:13}. The distinction of orthologous and paralogous gene
pairs is therefore of key practical importance for the functional
annotation of genomes. Most orthology assignment methods that are currently
in use for large data sets start from reciprocal best matches
\cite{Tatusov:97}, i.e., pairs of genes $x$ in species $A$ and $y$ in
species $B$ such that $y$ is the gene in $B$ most closely related to $x$
and $x$ is the gene in $A$ most closely related to $y$, see e.g.\
\cite{Altenhoff:09,Altenhoff:16,Setubal:18a}.  Reciprocal best matches are
efficiently computed in practice by quantifying sequence
similarity. Conceptually, they are employed because they approximate pairs
of reciprocal evolutionarily most closely related genes. It can be show
rigorously, that all pairs of orthologs are also reciprocal best matches in
the evolutionary sense \cite{GGL:19}.

Genes evolve along a gene tree $T$ from which best matches and reciprocal
best matches can be defined (see next section for precise definitions).
Reciprocal best match graphs (RBMGs) are vertex-colored graphs where the
vertices represent genes and the colors designate the species in which the
genes reside \cite{Geiss:19a}. RBMGs have recently been characterized by
Gei{\ss} et al.\ \cite{Geiss:19a, Geiss:19x}. In practical applications, however, estimates
of RBMGs are plagued with measurement errors and noise, and thus the
empirically inferred graphs usually violate the property of being an RBMG. A
natural remedy to reduce the measurement noise is of course to modify the
empirical graph to the closest RBMG. In the first part of this contribution
we therefore consider the computational complexity of modifying
vertex-colored graphs to RBMGs and show that these problems are NP-hard.

Due the importance of orthology, it also of interest to investigate those
RBMGs that completely describe orthology relationships rather than
containing the orthology relation as a subgraph. It is shown in
\cite{Geiss:19x} that the ``orthology RBMG'' are exactly the
\emph{hierarchically colored cographs} (\hc-cographs). These are cographs
\cite{Corneil:81} with a particular vertex coloring. As shown in
\cite{VGH+19}, every cograph admits a hierarchical coloring; more
precisely, every greedy coloring \cite{Christen:79} of a cograph is
hierarchical (but not \emph{vice versa}).

Simulation studies show that estimates of RBMGs are typically not
\hc-cographs \cite{GGL:19}, and thus do not directly describe orthology.
In fact, they usually feature both false positive and false negative
edges. It is of interest, therefore, to consider the computational problem
of modifying a given vertex colored graph to an \hc-cograph. In the setting
considered here, the assignment of genes to the species in which the occur
is perfectly known, hence the coloring of the vertices must not be
changed. We therefore stay within the realm of graph modification by 
insertion/deletion of edges. The vertex coloring only brings additional
constraints to the table.

Ignoring these additional constraints arising from the vertex coloring, the
analogous problem of editing empirical graphs to the nearest cographs was
used to extract phylogenetic information from empirical reciprocal best
matches in \cite{Hellmuth:15a}. The (uncolored) cograph editing problem is
known to be NP-complete \cite{LIU201245}. Here we show that the
  colored version remains NP-complete.

\section{Preliminaries}

\paragraph{\textbf{Basics}}

Throughout we consider undirected graphs $G=(V,E)$ with vertex set $V$ and
edge set $E \subseteq \binom{V}{2}$.  An edge $\{x,y\}$ between vertices
$x$ and $y$ will be arbitrarily denoted by $xy$ or $yx$. For a graph
$G=(V,E)$ and a subset $W \subseteq V$ we denote by $G[W]=(W,F)$ the
\emph{induced subgraph} of $G$ where $F\subseteq E$ and $xy\in F$ for all
$xy\in E$ with $x,y \in W$.

A \emph{vertex coloring} of $G$ is a surjective map $\sigma:V\to S$. We
will write $(G,\sigma)$ to indicate the vertex coloring $\sigma$ of $G$.  A
graph $G$ is \emph{properly colored} if $xy\in E$ implies
$\sigma(x)\neq \sigma(y)$. A \emph{hub-vertex} $x\in V$ is a vertex that is
adjacent to all vertices in $V\setminus \{x\}$.  Thus, a hub-vertex in a
properly colored graph $(G,\sigma)$ always satisfies
$\sigma(x)\neq\sigma(v)$ for any $v\in V\setminus \{x\}$.

For a graph $G=(V,E)$ and a vertex $x\in V$ we define the sets
$V-x \coloneqq V\setminus \{x\}$ and
$E-x \coloneqq E\setminus\{xv\mid v\in V\}$.  The graph
$G-x\coloneqq ( x,E-x)$ is, therefore, obtained from $G$ by removing vertex
$x$ and all its incident edges.  In addition, we define for a vertex
$x\notin V$ the sets $V+x \coloneqq V\cupdot \{x\}$ and
$E +x \coloneqq E\cupdot\{xv\mid v\in V\}$.  Thus, the graph
$G+x\coloneqq (V+x,E+x)$ is obtained from $G$ by the adding vertex $x$ and
all edges of the form $xv$, $v\in V$, making $x$ to a hub-vertex in $G+x$.
Moreover, we write $G\odot F \coloneqq (V,E\odot F)$, where
$\odot\in\{\setminus, \symdiff\}$ and $\setminus$, resp., $\symdiff$
denotes the usual set-difference, resp., symmetric difference of two sets.
Let $G=(V,E)$ and $H=(W,F)$ be two distinct graphs. We write
$G\cupdot H\coloneqq (V\cupdot W, E\cupdot F)$ for their \emph{disjoint
  union} and
$G\join H\coloneqq (V\cupdot W, E\cupdot F \cupdot \{xy\mid x\in V,y\in
W\})$ for their \emph{join}.

\paragraph{\textbf{Trees}}
A \emph{phylogenetic tree} $T=(V,E)$ (\emph{on $L$}) is a rooted tree with
root $\rho_T$, leaf set $L\subseteq V$ and inner vertices
$V^0=V\setminus L$ such that each inner vertex of $T$ (except possibly the
root) is of degree at least three.

\emph{Throughout this contribution, we assume that every tree is
  phylogenetic.}

The \emph{restriction} $T_{|L'}$ of a tree $T$ to a subset $L'\subseteq L$
of its leaves is the tree with leaf set $L'$ that is obtained from $T$ by
first taking the minimal subtree of $T$ with leaf set $L'$ and then
suppressing all vertices of degree two with the exception of the root
$\rho_{T_{|L'}}$. A \emph{star-tree} is a tree such that the root is
incident to leaves only, i.e, it is either the single vertex graph $K_1$ or
a tree where the root is a hub-vertex.

The last common ancestor $\lca_T(x,y)$ of two distinct leaves $x,y\in L$ is
the vertex that is farthest away from the root and that lies on both two
paths from $\rho_T$ to $x$ and from $\rho_T$ to $y$.  For vertices
$u,v\in V$ we write $u \preceq_T v$ if $v$ lies on the unique path from the
root to $u$.
 
Trees can be equipped with vertex labels. The inner vertex label is defined
as a map $t:V^0 \to \{0,1\}$. The leaf label is a surjective map
$\sigma:L\to S$. In out setting, the map $\sigma$ is used to assign to each
gene $u\in L$ the species $\sigma(u)\in S$ in which $u$ resides.  Moreover,
the labels $0$ and $1$ on $V^0$ indicate what type of mechanism caused a
divergence of lineages: $0$ represents gene duplications and $1$ designates
speciation events.

For a subset $L'\subseteq L$ we write
$\sigma(L')=\{\sigma(x)\mid x\in L'\}$.  Moreover, we use the notation
$\sigma_{|L'}$ for the surjective map $\sigma:L'\to \sigma(L')$.  We also
write $L[s]\coloneqq \{x \mid x\in L, \sigma(x)=s\}$ for the set of all
leaves with color $s$.  In addition we will use the notation $(T,t)$,
$(T,\sigma)$, resp., $(T,t, \sigma)$ to emphasize that the tree $T$ is
equipped with a vertex label $t$, $\sigma$, resp., both.

\paragraph{\textbf{Hierarchically Colored Cographs}}

A graph $G$ is a \emph{cograph} if either $G=K_1$ or $G$ is the disjoint union
$G=\bigcupdot_i G_i$ of two or more cographs $G_i$, or $G$ is the join
$G=\bigjoin_i G_i$ of two or more cographs $G_i$.  An important
characterization states that $G$ is a cograph if and only if it does not
contain a path $P_4$ on 4 vertices as an induced subgraph
\cite{Corneil:81}.  Here we are interested in particular in cographs with a
particular type of vertex coloring, so-called \emph{hierarchically colored
  cographs} \cite{Geiss:19x,VGH+19}. We first define the disjoint
  union and the join for vertex-colored graphs:
\begin{definition}
  Let $(H_1,\sigma_{H_1})$ and $(H_2,\sigma_{H_2})$ be two vertex-disjoint
  colored graphs.  Then
  $(H_1,\sigma_{H_1}) \join (H_2,\sigma_{H_2}) \coloneqq (H_1\join
  H_2,\sigma)$ and
  $(H_1,\sigma_{H_1}) \cupdot (H_2,\sigma_{H_2}) \coloneqq (H_1\cupdot
  H_2,\sigma)$ denotes their join and union, respectively, where
  $\sigma(x) = \sigma_{H_i}(x)$ for every $x\in V(H_i)$, $i\in\{1,2\}$.
\end{definition}

\begin{definition}[\hc-cograph]
  An undirected colored graph $(G,\sigma)$ is a \emph{hierarchically
    colored cograph (\hc-cograph)} if 
 \begin{description}
  \item[\AX{(K1)}] $(G,\sigma)=(K_1,\sigma)$, i.e., a colored vertex, or 
  \item[\AX{(K2)}] $(G,\sigma)= (H,\sigma_H) \join (H',\sigma_{H'})$ and
    $\sigma(V(H))\cap \sigma(V(H'))=\emptyset$, or 
  \item[\AX{(K3)}] $(G,\sigma)= (H,\sigma_H) \cupdot (H',\sigma_{H'})$ and
    $\sigma(V(H))\cap \sigma(V(H')) \in \{\sigma(V(H)),\sigma(V(H'))\}$,
\end{description}
where both $(H,\sigma_H)$ and $(H',\sigma_{H'})$ are \hc-cographs. 
\label{def:hc-cograph}
\end{definition}
\AX{(K2)} ensures that an \hc-cograph is always properly colored, cf.\
\cite[Lemma 43]{Geiss:19x}. Definition~\ref{def:hc-cograph} reduces to the
usual recursive definition of cographs when the coloring information is
ignored. Thus every \hc-cograph is a cograph. The converse, however, is not
true. To see this, consider the graph
$G=K_1\cupdot K_1=(\{x,y\},E=\emptyset)$ with the coloring
$\sigma(x)\neq \sigma(y)$. Although $G$ is clearly a cograph, it violates
Property \AX{(K2)}. However, for every cograph $G$ there is a vertex
  coloring $\sigma$ such that $(G,\sigma)$ is an \hc-cograph \cite{VGH+19}.

 To emphasize  the number $n$ of distinct colors used for the vertices in an \hc-cograph, we often speak explicitly
	of \emph{$n$-\hc-cographs.}

\paragraph{\textbf{Orthology and Reciprocal Best Matches}}

Reciprocal best matches are used in practice to estimate
  orthology. In the following paragraph we clarify the relationship
  between the two concepts to the extent need here. For a more in-depth
  discussion we refer to \cite{GGL:19}.

In the following, let $T=(V,E)$ be a tree on $L$ together with vertex
labeling $t:V^0 \to \{0,1\}$ and $\sigma:L\to S$. We distinguish here two
relationships (orthology and reciprocal best matches) that may hold between
pairs of vertices in $L$. Both relations are defined in terms of the
topology of $T$, however, the orthology relations is defined by means of
the label $t$ and reciprocal best matches are defined by means of the label
$\sigma$.
\begin{definition}[\cite{Fitch:00}]
  Two leaves $x,y\in L$ are \emph{orthologs} in $(T,t)$ if and only if
  $t(\lca_T(x,y))=1$. A graph $G$ is an orthology graph if there is a tree
  $(T,t)$ such that $xy\in E(G)$ if and only if $t(\lca_T(x,y))=1$.
\end{definition}
Orthology is a symmetric relation. It has been shown that the
  orthology graphs are exactly the cographs \cite[Cor.\ 4]{HHH+13}.

\begin{definition}
  The leaf $y$ is a \emph{best match} of the leaf $x$ in the tree
  $(T,\sigma)$ if and only if $\sigma(x)\neq \sigma(y)$ and
  $\lca_T(x, y) \preceq_T \lca_T(x, y')$ for all leaves $y'$ with
  $\sigma (y' ) = \sigma (y)$.  If $x$ is also a best match of $y$, we call
  $x$ and $y$ \emph{reciprocal best matches}.  The graph with vertex set
  $L$ that has precisely all reciprocal best matches of $T$ as its edge is
  denoted by $G(T,\sigma)$.  A properly vertex-colored graph $(G,\sigma)$
  is a \emph{Reciprocal Best Match Graph (RBMG)} if there is a leaf-labeled
  tree $(T,\sigma)$ such that $G(T,\sigma) = (G,\sigma)$.
\end{definition}
In other words, $y$ is a best match of $x$ if $y$ is the closest relative
of $x$ in comparison with all other leaves from species $\sigma(y)$, see
Fig.\ \ref{fig:exmpl} for an illustrative example. We say that $(G,\sigma)$
\emph{is explained} by $(T,\sigma)$ if $G(T,\sigma) = (G,\sigma)$.  To
emphasize the number of species, i.e., the number $n:=|\sigma(L)|$ of
distinct colors, we often speak explicitly of \emph{$n$-RBMGs}.

\begin{figure}
  \begin{center}
    \includegraphics[width=.7\textwidth]{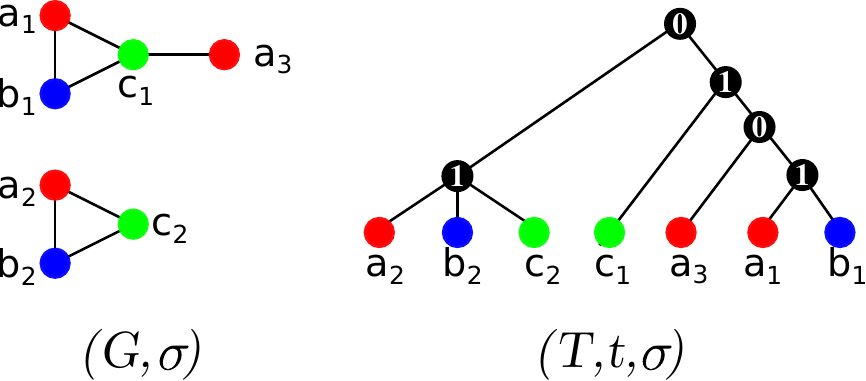}
  \end{center}
  \caption{The disconnected 3-RBMG $(G,\sigma)$ is also an orthology graph
    and thus, by Thm.\ \ref{thm:iffs}, an \hc-cograph.  For the tree
    $(T,t,\sigma)$ we have $G(T,\sigma) = (G,\sigma)$ as well as
    $t(\lca_T(x,y)) = 1$ iff $xy\in E(G)$. }
  \label{fig:exmpl}
\end{figure}

Characterizations of $2$-RBMGs, 3-RBMGs, and RBMGs that are orthology
graphs on $n$ colors have become available \cite{Geiss:19a,Geiss:19x} very
recently. For later reference we summarize these results in
\begin{theorem} Let $(G,\sigma)$ be a colored graph. Then, the following
  statements are satisfied:
\begin{enumerate}
\item $G$ is an orthology graph if and only if $G$ is a cograph.
\item The following statements are equivalent
  \begin{enumerate}
  \item $(G,\sigma)$ is an RBMG and an orthology graph.
  \item $(G,\sigma)$ is an RBMG and a cograph.
  \item $(G,\sigma)$ is an \hc-cograph.
  \end{enumerate}
\item $(G,\sigma)$ is a $2$-RBMG if and only if $(G,\sigma)$ is a properly
  $2$-colored bicluster graph that contains at least one edge.
\item $(G, \sigma)$ is an $n$-RBMG if and
  only if $(G, \sigma)$ is properly colored and each of its connected components is an RBMG and at least one
  connected component $C$ contains all colors, i.e., $|\sigma(V(C))| = n$.
\end{enumerate}
\label{thm:iffs}
\end{theorem}
\begin{proof}
  The first statement is equivalent to \cite[Cor.\ 4]{HHH+13}, the second
  statement is equivalent to \cite[Thm.\ 9]{Geiss:19x} and the last
  statement to \cite[Thm.\ 3]{Geiss:19x}.  For the third statement, note
  that \cite[Cor.\ 6]{Geiss:19a} states that $(G,\sigma)$ is a properly
  $2$-colored bicluster graph, whenever $(G,\sigma)$ is a $2$-RBMG.
  Together with Statement \AX{(4.)} this implies the only-if direction. For
  the if-direction observe that every properly colored biclique $H=(W,E)$
  (which may contain even only one vertex) is an RBMG since it is explained
  by a star-tree $(T_H,\sigma')$ on leaf set $W$.  Since $(G,\sigma)$
  contains an edge, it satisfies the conditions in Statement \AX{(4.)} and
  hence, is a $2$-RBMG.
\end{proof}

We note in passing, that $n$-RBMGs can be recognized in polynomial time for
$n\leq 3$.  However, although a mathematical characterization for
$n$-RBMGs, $n>3$ exists, it is still an open problem whether a
polynomial-time algorithm for their recognition exists \cite{Geiss:19x}.

Vertices in an RBMG $(G,\sigma)$ that have the same color induce, by
definition, an independent set in $G$. Moreover, since the definition of
$x$ and $y$ being reciprocal best matches does not depend on the presence
or absence of vertices $u$ with $\sigma(u)\notin\{\sigma(x),\sigma(y)\}$,
we have
\begin{fact}
  \label{obs-1}
  Let $(G,\sigma)$ be an RBMG explained by $T$ on $L$ and
  $L' := \bigcup_{s\in S'} L[s]$ be the subset of vertices with a
  restricted color set $S'\subseteq S$. Then the induced subgraph
  $(G[L'],\sigma_{|L'})$ is explained by the restriction $T_{|L'}$ of $T$
  to the leaf set $L'$ and thus, an RBMG.
\end{fact}

\paragraph{\textbf{Bicluster Graphs}}
As we shall see, there is a close connection between $2$-RBMGs and so-called
bicluster graphs.  For two sets $A$ and $B$, we write
$A\otimes B\coloneqq \{\{x,y\} \mid (x,y)\in A \times B\}$.  A
\emph{biclique} is a complete bipartite graph $G=(V,E)$, i.e., $G$ is
either the single vertex graph $K_1$ or a bipartite graph with bipartition
$V=V_1\cupdot V_2$ and edge set $E = V_1\otimes V_2$.  A \emph{bicluster
  graph} is a graph whose connected components are bicliques.

Bicluster graphs have been the subject of several studies, see e.g.\
\cite{Amit:04, CAI:96, DKJ+01,deSousaFilho2017, DRV:13, Guo+08, Ho:98,
  Peeters:03, Protti2009}. Often the problem is to modify a given graph
into a biclique or bicluster graph in some or the other way.  In this
contribution, we will utilize the so-called bicluster deletion and editing
problem. Amit \cite{Amit:04} showed that these problems are NP-complete.

\begin{problem}[\PROBLEM{Bicluster Deletion}]\ \\
  \begin{tabular}{ll}
    \emph{Input:}    & A bipartite graph $G =(V,E)$ and an integer $k$.\\
    \emph{Question:} & Is there a subset $F\subseteq E$ such that 
                       $G\setminus F$ is a bicluster graph 
                       and $|F|\leq k$?
  \end{tabular}
\end{problem}

\begin{problem}[\PROBLEM{Bicluster Editing}]\ \\
  \begin{tabular}{ll}
    \emph{Input:}    & A bipartite graph $G =(V,E)$ with bipartition
                       $V=V_1\cupdot V_2$  and an integer $k$.\\
    \emph{Question:} & Is there a subset $F\subseteq V_1\otimes V_2$ such that 
                       $G\symdiff F$ is a bicluster graph 
                       and $|F|\leq k$?
  \end{tabular}
\end{problem}

\begin{theorem}[\cite{Amit:04}]
    \PROBLEM{Bicluster Deletion} and \PROBLEM{Bicluster Editing} are NP-complete.
  \label{thm:biCl-NP}
\end{theorem}

The bicluster completion problem, which consists in finding the minimum
number of edges to add so that the resulting graph is a bicluster graph can
be solved in polynomial time.  To this end it is only necessary to identify
connected components and to add edges in each component to form a biclique.

\section{Complexity Results}

In the following we are interested in several problems that are
concerned with modifying an colored graph to RBMGs or \hc-cographs.  In
particular, we consider the following decision problems:

\begin{problem}[\PROBLEM{n-RBMG Deletion / n-\hc-cograph Deletion}]\ \\
  \begin{tabular}{ll}
    \emph{Input:}    & A properly $n$-colored graph $(G =(V,E),\sigma)$ 
                       and an integer $k$.\\
    \emph{Question:} & Is there a subset $F\subseteq E$ such that $|F|\leq
                       k$ and $(G\setminus F,\sigma)$ is an RBMG, resp., \\ &
                       \hc-cograph?
  \end{tabular}
\end{problem}

\begin{problem}[\PROBLEM{n-RBMG Editing / n-\hc-cograph Editing}]\ \\
  \begin{tabular}{ll}
    \emph{Input:}    & A properly $n$-colored graph $(G =(V,E),\sigma)$ 
                       and an integer $k$.\\
    \emph{Question:} & Is there a subset $F\subseteq \binom{V}{2}$ such
                       that $|F|\leq k$ and
                       $(G\symdiff F,\sigma)$ is an RBMG, resp., \\ & \hc-cograph? 
  \end{tabular}
\end{problem}

Note that the two problems \PROBLEM{n-RBMG Deletion} and \PROBLEM{n-\hc-cograph
  Deletion} are equivalent to the problem of finding a spanning subgraph
$(H,\sigma)$ of $(G,\sigma)$ with a maximum number of edges so that
$(H,\sigma)$ is an $n$-RBMG and an $n$-\hc-cograph, respectively.
 
\begin{remark}
  The input of the latter problems is a properly $n$-colored graph
  $(G =(V,E),\sigma)$.  Thus, for all edges $xy$ in $(G,\sigma)$, we have
  $\sigma(x)\neq\sigma(y)$.  Clearly, for $F\subseteq E$ we thus have
  $\sigma(x)\neq\sigma(y)$ for all edges $xy\in F$.  Moreover, for the two
  editing problems \PROBLEM{n-RBMG Editing} and \PROBLEM{n-\hc-cograph Editing},
  the graph $(G\symdiff F,\sigma)$ must in both cases be an RBMG (cf.\
  Thm.\ \ref{thm:iffs}).  This implies that $\sigma(x)\neq\sigma(y)$ for
  all edges $xy\in F$.

  In summary, an optimal deletion or edit set $F$ will never contain pairs
  $\{x,y\}$ with $\sigma(x) =\sigma(y)$.
  \label{obs:F-wellcolored}
\end{remark}

Note that if $(G,\sigma)$ is an $n$-colored but edge-less graph, then Thm.\
\ref{thm:iffs}(4) implies that $(G,\sigma)$ is not an RBMG.  In this case,
only the editing problem would be of interest for us.  However, this is a
trivial endeavor since an optimal edit set $F$ for $n$-colored edge-less
graphs always satisfies $|F| = \binom{n}{2}$. To see this, observe
first that we must connect the vertices such that at least one component
contains all colors. Moreover, by \cite[Cor.\ 2]{Geiss:19x}, there must be
an edge for every pair of distinct colors. Thus, at least $\binom{n}{2}$
edges must be added.  This can trivially be achieved by picking
$n$ distinctly colored vertices and connect them all by an edge
to obtain the modified graph $(H,\sigma)$.  Hence, if $(H,\sigma)$ is
disconnected, all its connected components are $K_1$s except the newly
created $n$-colored complete subgraph. Otherwise, $(H,\sigma)$ is an
$n$-colored complete graph $K_n$.  It is easy to see that $(H,\sigma)$ is
an \hc-cograph and thus, by Thm.\ \ref{thm:iffs} an RBMG.  Since at least
$\binom{n}{2}$ edges must be added, this editing is optimal and thus,
$|F| = \binom{n}{2}$.

\begin{remark}
  We will exclude the latter trivial case and will, from here on,
  \emph{always} assume that $(G,\sigma)$ contains at least one edge.
\label{rem:withEdges}
\end{remark}

\subsection{Graphs on two colors}

Since bicluster graphs do not contained induced $P_4$s, 
they are cographs. This together with Thm.\ \ref{thm:iffs} implies
\begin{corollary}
  The following statements are equivalent:
  \begin{enumerate}
  \item $(G,\sigma)$ is a $2$-RBMG. 
  \item $(G,\sigma)$ is a properly $2$-colored bicluster graph that
      contains at least one edge.
  \item $(G,\sigma)$ is a $2$-RBMG and an orthology graph.
  \item $(G,\sigma)$ is a $2$-\hc-cograph.
  \end{enumerate}
\label{cor:equiDefs}	
\end{corollary}

In the following we say that $F$ is an  \emph{(RBMG) deletion set} or an
  \emph{edit set}, resp., for a properly $n$-colored graph $(G,\sigma)$ if
  $(G-F,\sigma)$ or $(G\symdiff F,\sigma)$, resp., is an
  $n$-RBMG. Moreover, we say that a deletion or edit set $F$ is optimal if
$F$ has the smallest number of elements among all deletion or edit sets
$F'$ that yield an $n$-RBMG $(G\symdiff F',\sigma)$.

\begin{lemma}
  Let $(G,\sigma)$ be a bipartite graph whose vertices are colored with 2
  colors based on the bipartition $V_1,V_2$ of $V(G)$. Then, the following
  statements are satisfied:
\begin{enumerate}
\item If $F\subseteq V_1\otimes V_2$ is a minimum-sized edit set making
  $G\symdiff F$ to a bicluster graph, then $G\symdiff F$ contains at least
  one edge and $(G\symdiff F,\sigma)$ is a $2$-RBMG.
\item If $F\subseteq E$ is a minimum-sized deletion set making
  $G\setminus F$ to a bicluster graph, then $G\setminus F$ contains at
  least one edge and $(G\setminus F,\sigma)$ is a $2$-RBMG.
\end{enumerate}
\label{lem:1edge}
\end{lemma}
\begin{proof}
  By Remark \ref{rem:withEdges}, we assume that $(G=(V,E),\sigma)$ contains
  at least one edge.  Assume, for contradiction, that $G\setminus F$,
  resp., $G\symdiff F$ does not contain edges and thus, $F=E$ and
  $G\setminus F = G\symdiff F$. Let $F' = F\setminus \{xy\}$ with
  $xy\in E$. In this case, $G\setminus F' = G\symdiff F'$ contains the
  single edge $xy$ and hence, is a bicluster graph.  However, $|F'|<|F|$
  contradicts the optimality of $F$. Thus, $G\setminus F$ and
  $G\symdiff F$ must contain at least one edge.

  We continue with showing that $(G\symdiff F,\sigma)$ and
  $(G\setminus F,\sigma)$ are $2$-RBMGs.  First observe that $(G,\sigma)$ is
  properly $2$-colored since the vertices of $G$ are colored w.r.t.\ the
  bipartition $V_1,V_2$.  By the latter arguments, $(G\symdiff F,\sigma)$
  and $(G\setminus F,\sigma)$ contain at least one edge.  Moreover, by
  construction, $F$ contains only (non)edges between distinctly colored
  vertices. Hence, $G\symdiff F$ and $G\setminus F$ are properly
  $2$-colored bicluster graphs with at least one edge and thus, by Cor.\
  \ref{cor:equiDefs}, they are $2$-RBMGs.
\end{proof}

Taken the latter results, we can easily derive the following
\begin{corollary}
  \PROBLEM{2-RBMG Deletion}, \PROBLEM{2-\hc-Cograph Deletion}, 
  \PROBLEM{2-RBMG Editing}, \PROBLEM{2-\hc-Cograph Editing}, 
  are NP-complete.
\label{cor-$2$-NPc}  
\end{corollary}
\begin{proof}
  By Cor.\ \ref{cor:equiDefs}, the two problems \PROBLEM{2-RBMG Deletion}
  and \PROBLEM{2-\hc-Cograph Deletion} as well as the two problems
  \PROBLEM{2-RBMG Editing} and \PROBLEM{2-\hc-Cograph Editing} are
  equivalent.  Thus, it suffices to show that \PROBLEM{2-RBMG Deletion} and
  \PROBLEM{2-RBMG Editing} are NP-complete.  In order to verify that a
  properly $2$-colored graph is a $2$-RBMG, Cor.\ \ref{cor:equiDefs}
  implies that it suffices to check that it is a bicluster graph with at
  least one edge, a task that can clearly be done in polynomial time.
  Thus, \PROBLEM{2-RBMG Deletion} and \PROBLEM{2-RBMG Editing} are contained
  in NP.

  We proceed with showing the NP-hardness.  To this end let $G =(V,E)$ be
  an arbitrary instance of \PROBLEM{Bicluster Deletion}, resp.,
  \PROBLEM{Bicluster Editing}.  Thus, $G =(V,E)$ is a bipartite graph with
  partition $V_1\cupdot V_2$ of $G$.  Hence, we can establish a $2$-coloring
  $\sigma$ of $V$ w.r.t.\ the two set $V_1$ and $V_2$.

  If $F\subseteq E$ is a minimum-sized deletion set making $G\setminus F$
  to a bicluster graph, then Lemma \ref{lem:1edge} implies that
  $(G\setminus F,\sigma)$ is a $2$-RBMG.  Conversely, if $F\subseteq E$ is a
  minimum-sized deletion set making $(G\setminus F, \sigma)$ to a $2$-RBMG,
  then Cor.\ \ref{cor:equiDefs} implies that $G\setminus F$ is a bicluster
  graph.  This establishes the NP-hardness of \PROBLEM{2-RBMG Deletion} and
  \PROBLEM{2-\hc-Cograph Deletion}.

  Assume now that $F\subseteq V_1\otimes V_2$ is a minimum-sized edit set
  making $G\Delta F$ to a bicluster graph.  Then, by Lemma \ref{lem:1edge},
  $(G\Delta F, \sigma)$ is a $2$-RBMG.  Conversely, suppose that
  $F\subseteq \binom{V}{2}$ is an optimal edit set for $(G,\sigma)$.  By
  Remark \ref{obs:F-wellcolored}, the edit set $F$ will never contain pairs
  $\{x,y\}$ with $\sigma(x) =\sigma(y)$.  Hence,
  $F\subseteq V_1\otimes V_2$. This together with Cor.\ \ref{cor:equiDefs}
  implies that $(G',\sigma)$ is bicluster graph.  This establishes the
  NP-hardness of \PROBLEM{2-RBMG Editing} and \PROBLEM{2-\hc-Cograph
    Editing}.
\end{proof}

\subsection{Graphs with more than two colors}

Next, we will show that \PROBLEM{n-RBMG Deletion} and \PROBLEM{n-RBMG
  Editing} is NP-hard by employing Cor.\ \ref{cor-$2$-NPc}. To this end, we
stepwisely extend an instance $(G,\sigma)$ of \PROBLEM{2-RBMG Deletion /
  Editing} to an instance of \PROBLEM{n-RBMG Deletion / Editing} by adding
$n-2$ hub-vertices. This eventually allows us to show that an optimal
deletion, resp., edit set for $(G,\sigma)$ is also an optimal deletion,
resp., edit set for the constructed \PROBLEM{n-RBMG Editing} instance.

\begin{lemma}
  Let $n>1$. Then, $(G,\sigma)$ is an $(n-1)$-RBMG if and only if
  $(G+x,\sigma')$ with $\sigma'(v) = \sigma(v)$ and
  $\sigma'(x)\neq \sigma(v)$ for all $v\in V(G)$ is an $n$-RBMG.
\label{lem:G+x}
\end{lemma}
\begin{proof}
  Let $(G =(V,E),\sigma)$ be an $(n-1)$-RBMG. Hence, there is a tree
  $(T,\sigma)$ that explains $(G,\sigma)$.  Add $x$ as a new leaf to
  $(T,\sigma)$ such that $x$ is incident to the root $\rho$ of $T$, which
  results in the tree $(T^{x},\sigma')$. 
  To verify that $(G+x,\sigma')$ is an $n$-RBMG, it suffices to show that
  $(T^{x},\sigma')$ explains $(G+x,\sigma')$.  Since
  $(T,\sigma)=(T^{x}_{|V},\sigma'_{|V})$ and $\sigma'(x)\notin\sigma(V)$,
  the RBMG $(G,\sigma)$ is clearly explained by the restriction
  $(T^{x}_{|V},\sigma'_{|V})$.  Moreover, we have
  $\lca_{T^{x}}(x,v)=\lca_{T^{x}}(x,v')=\rho$ for all $v,v'\in V$.  This
  and $\sigma'(x)\notin\sigma(V)$ immediately implies $xv\in E(G(T^{x}))$
  for every $v\in V$. Hence, $(G+x,\sigma')=G(T^{x},\sigma')$, i.e.,
  $(T^{x},\sigma')$ explains $(G+x,\sigma')$. This and
  $|\sigma(V)|+1 = n = |\sigma'(V+x)|$ implies that $(G+x,\sigma')$ is an
  $n$-RBMG.
  
  Let $(G+x = (V,E),\sigma')$ be an $n$-RBMG where $\sigma(x)=r$ and let
  $(T,\sigma')$ be a tree on $L$ that explains $(G+x,\sigma')$.  Let
  $S'=S\setminus \{r\}$ and $L' := \bigcup_{s\in S'} L[s]$.  By Obs.\
  \ref{obs-1}, $((G+x)[L'],\sigma'_{|L'})$ is an $(n-1)$-RBMG.  By
  definition, $x$ is the only vertex in $G+x$ with color $r$ and thus
  $L' = V-x$. This in particular implies that $(G+x)[L'] = G$ and
  $\sigma(v)= \sigma'_{|L'}(v)$ for all $v\in V-x$.  Hence, $(G,\sigma)$ is
  an $(n-1)$-RBMG.
\end{proof}

\begin{lemma}
  Let $(G,\sigma)$ be a properly $n$-colored graph with hub-vertex $x$.
  Moreover, let $F$ be an optimal RBMG deletion, resp., edit set for $(G,\sigma)$.
	Then, $F$ does not contain any of the edges $xv\in E$.
  \label{lem:No-x-edits}
\end{lemma}
\begin{proof}
  Let $F$ be an optimal deletion, resp., edit set for $(G,\sigma)$. 
  Assume, for contradiction, that $F$ contains at least one edge $xv\in E$.
  Partition $F$ into a set $F_x$ that contains all edges of the form
  $xw\in F$ and $F_{\neg x}=F\setminus F_x$

  Now, put $H\coloneqq G\odot F$ with $\odot\in\{\setminus, \symdiff\}$.
  Thus, $(H,\sigma)$ is an $n$-RBMG that is explained by a tree
  $(T,\sigma)$ on $L$.  Let $S'=S\setminus \{r\}$, where $\sigma(x)=r$, and
  $L' := \bigcup_{s\in S'} L[s]$.  
  Since $x$ is a hub-vertex in $(G,\sigma)$, $x$ is the only vertex in
  $(G,\sigma)$ with color $r$ and hence, $L'=V-x$.
	This and Obs.\ \ref{obs-1} imply that 
  $(H[L'] ,\sigma_{|L'})$ is an $(n-1)$-RBMG and Lemma
  \ref{lem:opt-edit-set} implies that $(H[L']+x,\sigma)$ is an $n$-RBMG.

 Thus, by
  construction, $H[L']+x = G\odot F_{\neg x}$ and therefore,
  $(G\odot F_{\neg x},\sigma)$ is an $n$-RBMG.  However,
  $|F_{\neg x}|< |F_{\neg x}|+|F_x| =|F|$; contradicting the optimality
  of $F$.
\end{proof}

\begin{lemma}
  Let $(G = (V,E),\sigma)$ be a properly $n$-colored graph and suppose that
  $(G,\sigma)$ contains a hub-vertex $x$.  Let
  $(H,\sigma')\coloneqq(G-x, \sigma_{|V-x})$.  Then, $F$ is an optimal
  deletion or edit set with $(G\setminus F,\sigma)$ or
  $(G\symdiff F,\sigma)$ being an $n$-RBMG if and only if $F$ is an optimal
  set such that $(H\setminus F,\sigma')$, resp., $(H\symdiff F,\sigma')$ is
  an $(n-1)$-RBMG.
  \label{lem:opt-edit-set} 
\end{lemma}
\begin{proof}		
  First, assume that $(G,\sigma)$ is a properly $n$-colored graph with
  hub-vertex $x$.  Since $(G,\sigma)$ is properly colored, we have
  $\sigma(x)\neq \sigma(v)$ for all $v\in V-x$ by definition of a
  hub-vertex.  Hence, as $(G,\sigma)$ is properly $n$-colored, the graph
  $(H,\sigma')$ must be properly $(n-1)$-colored.  Now let $F$ be an
  optimal deletion, resp., edit set such that $(G\setminus F, \sigma)$,
  resp., $(G\symdiff F, \sigma)$ is an $n$-RBMG.  Assume, for
  contradiction, that $F$ is not optimal for $(H,\sigma')$. Thus, there
  exists a set $F'$ with $|F'|<|F|$ such that $(H',\sigma')$ is an
  $(n-1)$-RBMG, where $H' = H\setminus F'$, resp., $H' = H\symdiff F'$.  By
  Lemma \ref{lem:G+x}, the graph $(H'+x,\sigma)$ is an $n$-RBMG and in
  particular, $H'+x = G\setminus F'$, resp., $H'+x = G\symdiff F'$; a
  contradiction to the optimality of $F$ for $(G,\sigma)$.

  Now assume that $F$ is an optimal deletion or edit set for $(H,\sigma')$
  and let $\odot= \setminus$ or $\odot= \symdiff$, resp.  By construction,
  $H+x = G$ and thus, $(H\odot F)+x = G\odot F$.  Since
  $(H\odot F,\sigma')$ is an $(n-1)$-RBMG, we can apply Lemma \ref{lem:G+x}
  to conclude that $(G\odot F, \sigma)$ is an $n$-RBMG.  Thus, $F$ is a
  deletion, resp., edit set for $(G,\sigma)$.  It remains to show that $F$
  is an \emph{optimal} set for $(G,\sigma)$.  Assume, for contradiction,
  that there exists a optimal deletion, resp., edit set $F'$ for $G$ with
  $|F'|<|F|$ such that $G\odot F'$ is an $n$-RBMG.  Lemma
  \ref{lem:No-x-edits} implies that $F'$ does not contain edges $xv$.
  Hence, $x$ remains a hub-vertex in $G\odot F'$.  Thus, we can apply Lemma
  \ref{lem:G+x} to conclude that the $(n-1)$-colored induced subgraph
  $H' = (G\odot F')-x$ that contains all vertices of $G$ with color
  distinct from $\sigma(x)$, is an $(n-1)$-RBMG.  However, since $F'$ does
  not contain edges $xv$, we have
  $H' = (G\odot F')-x = (G-x)\odot F'= H\odot F'$.  Hence,
  $(H\odot F',\sigma')$ is an $(n-1)$-RBMG and $|F'|<|F |$; contradicting
  the optimality of $F$ for $(H,\sigma')$.
\end{proof}

Currently, there are no known polynomial-time algorithm to verify, for
arbitrary integers $n$, whether a given graph $(G,\sigma)$ is an $n$-RBMG
or not. Hence, we do not know whether \PROBLEM{n-RBMG Deletion} and \PROBLEM{n-RBMG
  Editing} are in NP or not.  Nevertheless, NP-hardness can easily be
shown.

\begin{theorem}
   \PROBLEM{n-RBMG Deletion} and \PROBLEM{n-RBMG Editing} is NP-hard.
\end{theorem}
\begin{proof}
  We prove this statement by induction on the number $n$ of colors.  Cor.\
  \ref{cor-$2$-NPc} implies that the base case $n=2$ is NP-complete.

  Now assume that \PROBLEM{(n-1)-RBMG Deletion / Editing} is NP-hard,
  $n\geq 2$.  Let $(H=(V',E'),\sigma')$ be an arbitrary instance of the
  \PROBLEM{(n-1)-RBMG Deletion / Editing} Problem.  Now, we construct an
  instance $(G,\sigma)$ of the \PROBLEM{n-RBMG Editing} as follows: Let
  $G=(V,E)$ be the graph obtained from $(H,\sigma')$ by adding a new vertex
  $x$ that is adjacent to all vertices of $H$, and thus, a hub-vertex in
  $(G,\sigma)$, where $\sigma(v) = \sigma'(v)$ and
  $\sigma(x)\neq \sigma'(v)$ for all $v\in V'$.  Thus,
  $(H,\sigma') = (G-x,\sigma_{|V-x})$ and we can apply Lemma
  \ref{lem:opt-edit-set} to conclude that $F$ is an optimal set such that
  $(G \setminus F,\sigma)$, resp., $(G \symdiff F,\sigma)$ is an $n$-RBMG
  if and only if $F$ is an optimal set such that $(H \setminus F,\sigma')$,
  resp., $(H \symdiff F,\sigma')$ is an $(n-1)$-RBMG, which completes the
  proof.
\end{proof}

So-far we have shown that \PROBLEM{n-RBMG Deletion} and \PROBLEM{n-RBMG Editing} is
NP-hard.  We continue with showing that the problems \PROBLEM{n-\hc-Cograph
  Deletion} and \PROBLEM{n-\hc-Cograph Editing} are NP-complete. The
  proofs are similar to the proofs for the NP-hardness of \PROBLEM{n-RBMG
  Deletion / Editing}.

\begin{lemma}
  Let $n>1$. Then, $(G,\sigma)$ is an $(n-1)$-\hc-cograph if and only if
  $(G+x,\sigma')$ with $\sigma'(v) = \sigma(v)$ and
  $\sigma'(x)\neq \sigma(v)$ for all $v\in V(G)$, is an $n$-\hc-cograph.
  \label{lem:n-1-coRBMG}
\end{lemma}
\begin{proof}
  By Thm.\ \ref{thm:iffs}, $(G,\sigma)$ is an \hc-cograph if and only it is
  an RBMG and a cograph. Thus, we can apply Lemma \ref{lem:G+x} to conclude
  that $(G,\sigma)$ is an $(n-1)$-RBMG if and only if $(G+x,\sigma')$ is an
  $n$-RBMG.  Hence, it remains to show that $G$ is a cograph if and only if
  $G+x$ is a cograph.  Note that $G_x = (V=\{x\}, E=\emptyset)$ is, by
  definition, a cograph. Again, by definition, $G$ is a cograph if and only
  if $G \join G_x = G+x$ is a cograph.
\end{proof}

\begin{lemma}
  Let $(G,\sigma)$ be a properly $n$-colored graph and suppose that
  $(G,\sigma)$ contains a hub-vertex $x$.  Let
  $(H,\sigma')=(G-x, \sigma_{|V-x})$.  Then, $F$ is an optimal deletion,
  resp., edit set such that $(G\setminus F,\sigma)$, resp.,
  $(G\symdiff F,\sigma)$ is an $n$-\hc-cograph if and only if $F$ is an
  optimal deletion, resp., edit set such that $(H\setminus F,\sigma')$,
  resp., $(H\symdiff F,\sigma')$ is an $(n-1)$-\hc-cograph.
  \label{lem:opt-edit-set-cograph}
\end{lemma}
\begin{proof}
  We show first that an optimal \hc-cograph deletion or edit set $F$ for
  $(G,\sigma)$ does not contain edges $xv$.  Assume, for contradiction,
  that $F$ contains at least one edge $xv\in E$ and let $F_{\neg x}$ be the
  set of edges in $F$ that are not incident to $x$.  Similar arguments as
  in the proof of Lemma \ref{lem:No-x-edits} together with Lemma
  \ref{lem:n-1-coRBMG} show that $ (G\odot F_{\neg x}, \sigma)$ is an
  \hc-cograph and thus, $|F_{\neg x}|< |F|$; a contradiction.

  Now we can reuse the arguments of the proof of Lemma
  \ref{lem:opt-edit-set} by utilizing Lemma \ref{lem:n-1-coRBMG} instead of
  Lemma \ref{lem:G+x}, which completes the proof of this lemma.
\end{proof}

\begin{theorem}
  The problems \PROBLEM{n-\hc-Cograph Deletion} and \PROBLEM{n-\hc-Cograph Editing}
  are NP-complete.
\end{theorem}
\begin{proof}
  \cite[Thm.\ 11]{Geiss:19x} shows that it can be verified in polynomial
  time whether a given colored graph $(G,\sigma)$ is an $n$-\hc-cograph or
  not. Thus, \PROBLEM{n-\hc-Cograph Deletion} and \PROBLEM{n-\hc-Cograph Editing}
  are contained in the class NP.

  The proofs that \PROBLEM{n-\hc-Cograph Deletion} and \PROBLEM{n-\hc-Cograph
    Editing} are NP-hard work exactly in the same way as the proof for
  showing the NP-hardness of \PROBLEM{n-RBMG Deletion} and \PROBLEM{n-RBMG
    Editing}.  Simply replace all $k$-RBMG by $k$-\hc-cograph in all of the
  desired steps and use Lemma \ref{lem:opt-edit-set-cograph} instead of
  Lemma \ref{lem:opt-edit-set}.
\end{proof}

As a consequence of Theorem \ref{thm:iffs}, we obtain
\begin{corollary}
  Let $(G,\sigma)$ be a properly $n$-colored graph and $k$ be an integer.
  Deciding whether there is a set $F \subseteq \binom{V}{2}$ of size
  $|F|\leq k$ such that $(G\setminus F,\sigma)$, resp.,
  $(G\symdiff F,\sigma)$ is an RBMG and a cograph is NP-complete.
\end{corollary}

The property of being an RBMG or an \hc-cograph is not hereditary.  As an
example consider the disconnected \hc-cograph $(G,\sigma)$ in Fig.\
\ref{fig:exmpl}.  Now take vertex $a_1$ and $b_2$. We have
$\sigma(a_1)\neq \sigma(b_2)$ and as an induced subgraph
$(G[a_1,b_2],\sigma_{|a_1b_2})$ a $2$-colored but edge-less graph.  By
Thm.\ \ref{thm:iffs}(2) and (4), $(G[a_1,b_2],\sigma_{|a_1b_2})$ can
neither be an \hc-cograph nor an RBMG. As a consequence, one cannot use
  the general results on the complexity of graph modification problems for
  general hereditary graph classes outlined e.g.\ in \cite{CAI:96}.

For the special case of the \PROBLEM{2-RBMG Editing} problem, however, we
can use established results since \PROBLEM{2-RBMG Editing} and
\PROBLEM{Bicluster Editing} are equivalent as long as the input graph
contains at least one edge (cf.\ Remark \ref{rem:withEdges} and Lemma
\ref{lem:1edge}).  For the \PROBLEM{Bicluster Editing} Problem,
Amit \cite{Amit:04} gave a factor-11 approximation.  Protti et al.\ \cite{Protti2009} showed
that \PROBLEM{Bicluster Editing} problem with input $(G=(V,E),\sigma)$ and
integer $k$ is FPT, and can be solved in $O(4^k+|V|+|E|)$ time by a
standard search tree algorithm.  Moreover, they showed how to construct a
problem kernel with $O(k^2)$ vertices in $O(|V|+|E|)$ time. Guo et al.\ \cite{Guo+08}
improved the latter results to a problem kernel with $O(k)$ vertices and an
FPT-algorithm with running time $O(3.24^k + |E|)$. These results together
with Cor.\ \ref{cor:equiDefs} imply
\begin{theorem} \PROBLEM{2-RBMG Editing} with input $(G,\sigma)$ and integer
  $k$, has a problem kernel with $O(k)$ vertices and an FPT-algorithm with
  running time $O(3.24^k + |E|)$.
\end{theorem}

Moreover, Guo et al.\ \cite{Guo+08} provided a randomized 4-approximation algorithm
improving the factor-11 approximation Amit \cite{Amit:04}.  Due to its
importance to integrate and analyze high-dimensional biological data on a
large scale many heuristics for the \PROBLEM{Bicluster Editing} problem
have been established in the last few years
\cite{deSousaFilho2017,SRP+14,de2012hybrid,GRS+15,puleo2018correlation,Chen12253,cheng2000biclustering,BUSYGIN20082964,LSH+10,TBK:05},
which can directly be applied for the \PROBLEM{2-RBMG Editing}.  For the
\PROBLEM{n-RBMG Editing} and the \PROBLEM{n-\hc-Cograph Editing} problem
for an arbitrary number of colors $n$, there are, to our knowledge, no
heuristics nor parameterized algorithms available so-far.

\section{Summary and Outlook}

We have shown that the four problems \PROBLEM{n-RBMG Editing},
\PROBLEM{n-\hc-Cograph Editing}, \PROBLEM{n-RBMG Deletion} and
\PROBLEM{n-\hc-Cograph Deletion} are NP-hard.  In addition, the two
\hc-cograph modification problems are NP-complete, and \hc-cograph
modification is equivalent to modifying a given graph into an RBMG
that represents an orthology relation.

We are left with some of open problems.  Since the latter four modification
problems are NP-hard, it is necessary to design efficient heuristics or
parameterized algorithms in order to correct graphs inferred from sequence
data to RBMGs or \hc-cographs. Although, we obtained an FPT-algorithm for
\PROBLEM{2-RBMG Editing} as a trivial by-product, so far no further
  results are available for a general number $n$ of colors.
 
Moreover, it is not known whether \PROBLEM{n-RBMG Editing} for $n\geq 4$ is
in NP or not.  While $2$- and $3$-RBMGs can be recognized in polynomial
time \cite{Geiss:19x}, it is not known whether there exists polynomial-time
algorithms for the recognition of $n$-RBMGs with $n\geq 4$.  In addition,
the complexity of the $n$-RBMG Completion Problem (i.e., adding a minimum
number of edges to obtain an $n$-RBMG) remains unsolved. We emphasize
  that this problem is not solved by adding to each $2$-colored induced
  subgraphs the minimum number of edges required to turn it into a
  bicluster graph, see \cite[Fig.\ 8(A)]{Geiss:19x} for a counterexample.

Furthermore, we observed in \cite{GGL:19} that for many of the estimated
RBMGs the quotient graph w.r.t.\ the colored thinness relation (which
  identifies vertices of the same color in $(G,\sigma)$ that have the same
  color and the same neighborhood) are $P_4$-sparse, i.e., each of its
induced subgraphs on five vertices contains at most one induced $P_4$
\cite{Jamison:92}.  It is well-known that uncolored $P_4$-sparse graphs can
be optimally edited to cographs in linear time \cite{LIU201245}.  This begs
the questions whether it is also possible to edit a $P_4$-sparse RBMG
  to an \hc-cograph in polynomial time.  More generally, can one
efficiently edit a not necessarily $P_4$-sparse or a weighted $P_4$-sparse RBMG to an \hc-cograph?

In this contribution we have considered graph modification problems
  for RBMGs. RBMGs are the symmetric part of the Best Match Graphs (BMGs)
  \cite{Geiss:19a}. In practical applications BMGs are initially estimated
  from sequence data and then processed to extract (approximate)
  RBMGs. Thus is seems natural to consider the corresponding digraph
  modification problems of BMGs. At present, this remains an open problem.

\section*{Acknowledgments}
This work was support in part by the German Federal Ministry of Education
and Research (BMBF, project no.\ 031A538A, de.NBI-RBC).

\bibliographystyle{plain}
\bibliography{hc-co}


\end{document}